\documentclass[
 reprint,
 amsmath,amssymb,
 aps,
]{revtex4-2}

\usepackage{graphicx}
\usepackage{hyperref}
\usepackage[linesnumbered, ruled, vlined, titlenotnumbered]{algorithm2e}
\usepackage{amsmath, amsfonts, amssymb, amsthm, color, bbold}
\usepackage{bm}
\usepackage[capitalise]{cleveref}

\usepackage[linesnumbered, ruled, vlined, titlenotnumbered]{algorithm2e}

\newtheorem{theorem}{Theorem}

\newtheorem{proposition}{Proposition}

\newtheorem{lemma}{Lemma}

\newcommand{\F}{{\mathbb{F}}}
\newcommand{\Ftwo}{\F_{\!2}}

\newcommand{\FQ}{\F_{\!Q}}

\newcommand{\x}[1]{x^{(#1)}}
\newcommand{\y}[1]{y^{(#1)}}
\newcommand{\z}[1]{z^{(#1)}}
\newcommand{\xmatrix}{{\cal X}}
\newcommand{\ymatrix}{{\cal Y}}
\newcommand{\zmatrix}{{\cal Z}}

\DeclareMathOperator{\rank}{rank}
\newcommand{\matrixset}[2]{\Ftwo^{#1 \times #2}}

\newcommand{\bx}{\mathbf x}

\newcommand{\balpha}{\bm\alpha}
\newcommand{\bbeta}{\bm\beta}
\newcommand{\bgamma}{\bm\gamma}
\renewcommand{\leq}{\leqslant}

\def\inner#1#2{\langle #1 , #2\rangle}
\def\gab#1{\mathrm{Gab}(#1)}
\def\qgab#1{\mathrm{QGab}(#1)}
\def\tr#1{{\mathrm Tr}{(#1)}}

\begin{document}

\title{Correction of circuit faults in a stacked quantum memory using rank-metric codes}

\author{Nicolas Delfosse}
\affiliation{IonQ Inc.}
\author{Gilles Z\'emor}
\affiliation{Institut de math\'ematiques de Bordeaux\\ Institut universitaire de France}

\date{\today}

\begin{abstract}
We introduce a model for a stacked quantum memory made with multi-qubit cells, inspired by multi-level flash cells in classical solid-state drive, and we design quantum error correction codes for this model by generalizing rank-metric codes to the quantum setting.
Rank-metric codes are used to correct faulty links in classical communication networks.
We propose a quantum generalization of Gabidulin codes, which is one of the most popular family of rank-metric codes, and we design a protocol to correct faults in Clifford circuits applied to a stacked quantum memory based on these codes.
We envision potential applications to the optimization of stabilizer states and magic states factories, and to variational quantum algorithms.
Further work is needed to make this protocol practical.
It requires a hardware platform capable of hosting multi-qubit cells with low crosstalk between cells, a fault-tolerant syndrome extraction circuit for rank-metric codes and an associated efficient decoder.
\end{abstract}

\maketitle
\section{Introduction}

To reach large-scale applications a quantum computer must be built around a quantum error correction scheme, responsible for the correction of faults occurring during the computation~\cite{shor1996fault}.

In this work, we consider a model for quantum computation based on a stacked memory with $\ell$ layers of $n$ qubits, represented on \cref{fig:stacked_memory}.
Qubits are grouped in cells containing $\ell$ qubits. 
We use the term cell to emphasize the resemblance with multi-level flash memories, widely adopted for classical storage, which encode multiple bits per cell~\cite{jung1996117, gao2012innovative, cai2017error}.
In the quantum setting, one could consider designing a cell using the energy levels of an atom or a quantum harmonic oscillator.
Qudits with 5, 7, and 13 levels have been demonstrated experimentally with ions~\cite{hrmo2023native, ringbauer2022universal, low2023control}.
High fidelity 3-level and 4-level systems have been realized recently using harmonic oscillators~\cite{brock2024quantum}.
Alternatively, a cell could be an entire module of a quantum computer, connected to other modules through interconnects.

A quantum circuit is executed simultaneously on each layer of the stacked memory, with potentially distinct inputs.
Stacked memories could find applications in fault-tolerant quantum computing, to improve the throughput of the factories used to produce many copies of quantum states consumed during the computation, {\em e.g.} logical zero states, logical plus states or logical Bell states consumed by Steane-style or Knill-style error correction~\cite{steane1997active, knill2005quantum, brun2018efficient} or magic states used to implement logical gates~\cite{bravyi2005universal}.

Below, we design a stacked version of the standard circuit noise model, incorporating the fact that a fault on a qubit is likely to affect other qubits in the same cell.
To correct these faults, we encode together the $\ell n$ qubits of the stacked memory.

\begin{figure}
    \centering
    \includegraphics[width=.9\linewidth]{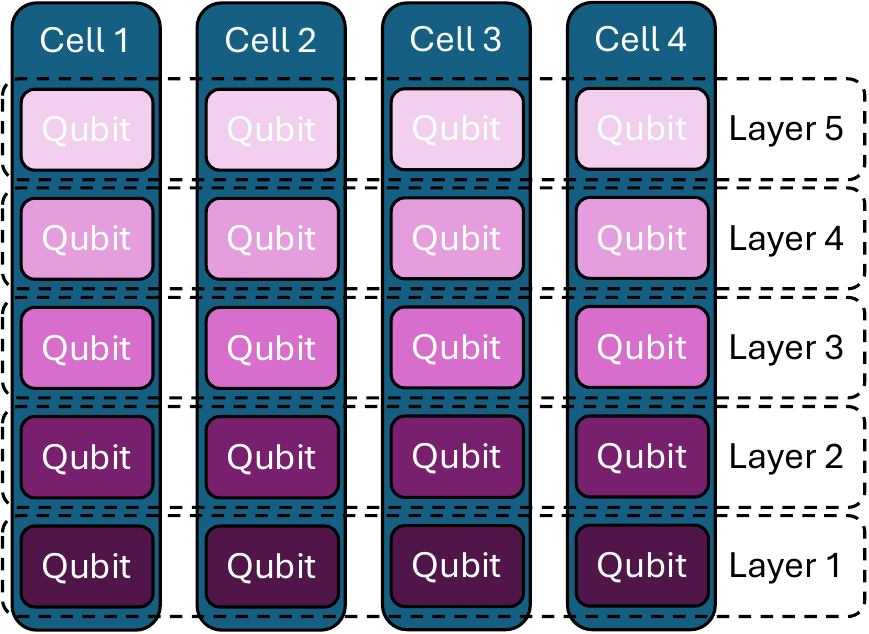}
    \caption{Abstract representation of a $5 \times 4$ stacked memory.}
    \label{fig:stacked_memory}
\end{figure}

The first ingredient of our protocol is \cref{lemma:w_stacked_faults_induce_rank_4w_change} proving that $w$ faults occurring during the stacked implementation of a Clifford circuit result in a Pauli error on the stacked memory which can be interpreted as a matrix with rank at most $4w$.
This key lemma suggests encoding the stacked memory in such a way that low-rank errors can be detected. 
This is precisely what classical rank-metric codes are designed for.
The second ingredient is the introduction of the quantum version of a popular family of classical rank-metric codes, the so-called Gabidulin codes, and the computation of their parameters.

In the remainder of this paper, \cref{sec:toy_model,sec:network_faults} introduce a classical version of our protocol, making the rest of the paper easier to follow, and \cref{sec:gabidulin_codes} reviews classical Gabidulin codes.
\cref{sec:stacked_implementation} introduces the stacked memory model.
Quantum Gabidulin codes are introduced in \cref{sec:quantum_gabidulin_codes} and they are applied to the correction of faults in the stacked implementation.
In conclusion, we discuss the main limitations of our protocol, related experimental results, and its potential applications.

\section{Toy model inspired by classical network coding}
\label{sec:toy_model}

The {\em network coding} paradigm considers information transmission through a network, from an input node to one or several output nodes, where inner nodes transmit linear combinations of their received packets through outgoing edges \cite{ahlswede2000network,li2003linear,
koetter2003algebraic}.
Here we adopt the point of view that the network represents multiplication by a binary matrix.

Formally, a {\em network} is defined to be a directed acyclic graph $G = (V, E)$.
The {\em inputs} (respectively {\em outputs}) vertices of the network are the vertices with no incoming (respectively outgoing) edge.
For simplicity, we assume that the sets of input and output vertices are disjoint.

With each edge $e = (v, w)$ of the network, is associated a $\Ftwo$-linear form $\ell_e: \Ftwo^{\delta} \rightarrow \Ftwo$ where $\delta=1$ if $v$ is an input vertex and $\delta$ is the number of incoming edges of $v$ otherwise.
This linear form represents the information sent through edge $e$.
Vertex $v$ receives $\delta$ bits from its incoming edges, applies $\ell_e$ to this $\delta$-bit vector, and sends the resulting bit to $w$ through edge $e$.
In addition to the linear forms associated with edges, a linear map $\Ftwo^{\delta} \rightarrow \Ftwo$ is associated with each output vertex $v$, where $\delta$ is the number of incoming edges of $v$.
We refer to the transmission over all the outgoing edges of $v$, after evaluating the corresponding linear forms, as the {\em processing of $v$}.

For simplicity, assume that the network has $n$ input vertices and $n$ output vertices.
It represents the map
$x \mapsto y = A x$
applying a $n \times n$ binary matrix $A$ to a vector $x \in \Ftwo^n$.
We treat $x$ and $y$ as column vectors.
The bits of $x$ are fed to the input vertices. The matrix $A$ is applied by executing the processing of all the vertices, sorted in topological order.
The result $y$ is emitted by the output vertices.
The topological order guarantees that each vertex has received all its input bits before its processing is executed.
Recall that any directed acyclic graph admits a topological order.

Our goal is to use the network to apply the matrix $A$ to many input vectors. The network includes faulty edges in unknown locations. If an edge is {\em faulty}, any bit sent through this edge has a probability $p$ to be flipped.

\section{Network faults}
\label{sec:network_faults}

The relevance of rank-metric codes to error protection in network coding schemes was first shown in \cite{silva2008rank}. In our case,
consider the effect of faulty edges:
assume that the network is used to apply $A$ to $m$ input vectors $\x{1}, \dots \x{m} \in \Ftwo^{n}$.
Form the $n \times m$ matrix $\xmatrix$ whose columns are the vectors $\x{i}$ and let $\ymatrix$ be the $n \times m$ matrix whose columns are the $\y{i} = A \x{i}$.
These matrices satisfy $\ymatrix = A \xmatrix$.
Because of faults, we obtain a result $\z{i}$, which is a random variable, instead of $\y{i}$.
Denote by $\zmatrix$ the corresponding matrix.

\begin{lemma}
\label{lemma:w_network_faults_induce_rank_w_change}
    If $t$ edges of the network are faulty, then we have $\rank(\ymatrix - \zmatrix) \leq t$.
\end{lemma}

\begin{proof}
By construction, the matrix $A$ can be decomposed as a product $A = A_s \dots A_2 A_1$ where $A_j$ is a matrix representing the processing of the $j$ th vertex $v_j$ in topological order.
Denote $\overrightarrow{A_{j}} = A_j \dots A_2 A_1$ and $\overleftarrow{A_{j}} = A_s \dots A_{j+2} A_{j+1}$.
The vector $\overrightarrow{A_{j}} \x{i}$
is the result obtained after processing the first $j$ vertices of the network with input $\x{i}$.

If one of the outgoing edges $e$ of $v_j$ is faulty, the bit of $\overrightarrow{A_{j}} \x{i}$  corresponding to $e$ may be flipped, resulting in $\overrightarrow{A_{j}} \x{i} + \varepsilon_e$ where $\varepsilon_e$ is a weight-1 vector.
Denote by $\overleftarrow{\varepsilon_{e}} \in \Ftwo^n$ the column vector $\overleftarrow{A_{j}} \varepsilon_e$.
It represents the error induced by $\varepsilon_e$ at the output of the network. 
Define the vector $\delta_e \in \Ftwo^{m}$ such that  $\delta_{e, i} = 1$ if a fault occurs on edge $e$ with input $\x{i}$ and $\delta_{e, i} = 0$ otherwise.
The output obtained from $\x{i}$ is
$
\z{i} = \y{i} + \overleftarrow{\varepsilon_{e}} \delta_{e, i}.
$
Therefore, the output matrix is 
\begin{align}
\zmatrix = \ymatrix + \overleftarrow{\varepsilon_{e}} \delta_e \cdot
\end{align}
Therein, $\overleftarrow{\varepsilon_{e}}$ is a column vector with length $n$ and $\delta_e$ is a row vector with length $m$.
Their product is a $n \times m$ matrix with rank at most one.
By linearity, faults on other edges have the a similar effect.
Let $E_f \subset E$ be the subset of faulty edges.
Then, the output matrix is of the form 
\begin{align}
\zmatrix = \ymatrix + \sum_{e \in E_f} \overleftarrow{\varepsilon_{e}} \delta_e \cdot
\end{align}
where $\overleftarrow{\varepsilon_{e}}$ and $\delta_e$ are defined like in the previous case.
The matrix $\zmatrix - \ymatrix$ is a sum of $|E_f|$ matrices with rank at most 1, which implies
$\rank(\ymatrix - \zmatrix) \leq |E_f|$, proving the result.
\end{proof}

\section{Correction of network faults with Gabidulin codes}
\label{sec:gabidulin_codes}

\cref{lemma:w_network_faults_induce_rank_w_change} motivates the introduction of rank-metric codes.
Intuitively, if the matrix $\ymatrix$ belongs to a set of matrices well separated from each other in rank-distance we should be able to correct $\zmatrix$ and to recover $\ymatrix$.

Let $n, k$ be two integers with $k \leq n$ and $Q = 2^n$.
Let $\balpha = [\alpha_1,\ldots,\alpha_n]$ be a basis of $\FQ$ viewed as an $\Ftwo$-linear space. 
The associated {\em Gabidulin code} $\gab{\balpha,k}$ is defined as the space of
vectors 
$
[f(\alpha_1),\ldots,f(\alpha_n)]
$ of $\FQ^n$ 
where $f$ ranges over the space of polynomials of the form
\begin{equation}\label{eq:qpoly}
f(X)=a_0X+a_1X^2+\cdots +a_iX^{2^i}+\cdots +a_{k-1}X^{2^{k-1}}
\end{equation}
with $a_i \in \FQ$, which are evaluated at the elements $\alpha_1,\ldots,\alpha_n$ of $\FQ$.
Gabidulin codes were first introduced by Delsarte \cite{delsarte1978} and made popular by Gabidulin \cite{gabidulin1985}.

Codewords can be seen as $n \times n$ binary matrices because any element of $\FQ$ can be regarded as a length-$n$ column vector, using the $\Ftwo$-linear space structure of $\FQ$.
Therefore, we can talk about the rank of a vector of $\FQ^n$ which is its rank as a $n \times n$ binary matrix.
The {\em minimum rank distance} of a code over $\F_Q$ is defined to be the minimum rank of a non-zero codeword.
The following standard result \cite{delsarte1978,gabidulin1985} is proven in \cref{appendix:proof_gab_code_parameters}.

\begin{theorem}
\label{theorem:gabidulin_codes_parameters}
The code $\gab{\balpha,k}$ is a $\FQ$-linear code with length $n$ and dimension $k$ over $\FQ$ and minimum rank distance equal to $n-k+1$.
\end{theorem}

\cref{theorem:gabidulin_codes_parameters} provides a strategy to correct network faults during the application of a matrix $A \in \matrixset{n}{n}$.
For simplicity, assume that $A$ is invertible.
We encode together $k$ input vectors using a Gabidulin code with length $n$ and minimum rank distance $n-k+1$.
The input vectors form a matrix $\xmatrix \in \matrixset{n}{k}$, encoded into $\bar \xmatrix \in \matrixset{n}{n}$, that is sent through the network.
The encoded matrix is obtained by adding redundant columns to $\xmatrix$.

In the absence of faults, the output of the network is the matrix $\bar \ymatrix = A\bar \xmatrix$.
Because $A$ is invertible, the set of all possible matrices $A\bar \xmatrix \in \matrixset{n}{n}$ also forms a code with minimum rank distance $n-k+1$, that we refer to as the {\em image code}.
We assume that it is equipped with an efficient decoder. The decoder takes as an input a matrix $\bar\zmatrix$ and it returns a matrix $\bar\ymatrix$ of the image code minimizing $\rank(\bar \ymatrix - \bar \zmatrix)$.

If $t$ faults occurs, the network outputs a matrix $\bar \zmatrix$ satisfying $\rank(\bar \ymatrix - \bar \zmatrix) \leq t$.
Applying the image code decoder, we recover $\bar \ymatrix$ from $\bar \zmatrix$ when $t \leq (n-k+1)/2$.
The matrix $\ymatrix = A\xmatrix$ is obtained by discarding the redundant columns of $\bar \ymatrix$.

This application is somewhat artificial because decoding the image code is generally non-trivial and it could be more difficult than the application of the matrix $A$.
However, this toy model is a steppingstone to the problem of correcting faults in Clifford circuits using quantum Gabidulin codes.

\section{Stacked implementation of a quantum circuit}
\label{sec:stacked_implementation}

Consider a $n$-qubit Clifford circuit $C$ with size $s$, made with unitary single-qubit and two-qubit gates.
Recall that if $P$ is a Pauli operator and $U$ is a Clifford gate, $U P U^{\dagger}$ is also a Pauli operator.

In the remainder of this paper, we introduce a stacked implementation which executes the circuit $C$ multiple times in parallel and corrects circuit faults using a quantum generalization of rank-metric codes introduced below.

A {\em $\ell \times n$ stacked memory}, represented in \cref{fig:stacked_memory}, is a register of $\ell n$ qubits stored in $n$ {\em cells} containing $\ell$ qubits each.
We refer to the $n$ qubits obtained by selecting the $i$ th qubit of each cell as the $i$ th {\em layer} of the stacked memory.

A {\em stacked implementation} with $\ell$ layers of an $n$-qubit circuit $C$ is defined to be the circuit obtained applying $C$ to each layer of a $\ell \times n$ stacked memory.
If the circuit $C$ is made of the gates $U_1, \dots, U_s$, the $\ell$-layer implementation of $C$ is made of the gates $U_1^{\otimes \ell}, \dots, U_s^{\otimes \ell}$.
Any single-qubit gate $U$ of $C$ supported on qubit $i$ is replaced by the gate $U^{\otimes \ell}$ supported on cell $i$.
Similarly any two-qubit gate $U$ acting on qubit $i$ and $j$ becomes the gate $U^{\otimes \ell}$ supported on cell $i$ and $j$.

We consider a stacked version of the standard circuit noise model.
Each gate $U^{\otimes \ell}$ of the stacked implementation is followed by a fault $P$ with probability $p$. The fault $P$ is selected uniformly among the non-trivial Pauli error acting on the cells supporting the gate.

A {\em stacked error} is defined to be a $\ell \times n$ matrix $P$ with coefficients in $\{I, X, Y, Z\}$, where $P_{i,j}$ represent the error affecting the qubit on layer $i$ of cell $j$.
Denote by $P_{i, \cdot}$ the restriction of $P$ to layer $i$.

The {\em rank} of a stacked error $P$, denoted $\rank(P)$, is defined to be the rank of the group generated by the operators $P_{i, \cdot}$ for $i=1,\dots,\ell$.

\begin{lemma}
\label{lemma:w_stacked_faults_induce_rank_4w_change}
If $t$ gates of the stacked implementation are faulty, then the state of the stacked memory at the end of the implementation suffers from a stacked error $Q$ with $\rank(Q) \leq 4t$.
\end{lemma}

\begin{proof}
The stacked memory undergoes the operation
\begin{align}
\label{eq:stacked_implementation_with_faults}
P^{(s)} U_s^{\otimes \ell} \dots  P^{(2)} U_2^{\otimes \ell}  P^{(1)} U_1^{\otimes \ell}
\end{align}
where $P^{(i)}$ is the fault following the gate $U_i^{\otimes \ell}$.
To obtain the effect of faults on the output of the circuit, we use the relation
\begin{align}
\label{eq:conjugating_stacked_faults}
U^{\otimes \ell} P
= P' U^{\otimes \ell}
\end{align}
where $P$ is a Pauli error occurring before the gate $U^{\otimes \ell}$ and $P'$ is obtained by conjugating each row of $P$ by $U$.
Applying \cref{eq:conjugating_stacked_faults}, we can move all the Pauli errors in \cref{eq:stacked_implementation_with_faults} to the left, which yields
\begin{align}
\label{eq:stacked_implementation_with_faults_propagated}
\lambda Q^{(s)} \dots Q^{(2)} Q^{(1)} U_s^{\otimes \ell} \dots U_2^{\otimes \ell} U_1^{\otimes \ell} \cdot
\end{align}
Therein,
$\lambda$ is a global phase that can be ignored and $Q^{(t)}$ is obtained by conjugating the rows of $P^{(t)}$
by $U_{s} \dots U_{t+1}$.
The product of the $Q^{(t)}$ is taken component by component.

Because we consider circuits $C$ containing only single-qubit gates and two-qubit gates, the stacked error $P^{(t)}$ following such a gate has rank at most $4$.
The conjugation being an automorphism of the Pauli group, it preserves the rank of $P^{(t)}$, meaning that $\rank(Q^{(t)}) = \rank(P^{(t)})$ for all $t$.
Moreover, we have 
\begin{align}
    \rank\left(\prod_{t=1}^s Q^{(t)}\right) \leq \sum_{t=1}^s \rank(Q^{(t)})
\end{align}
which is at most four times the number of non-trivial faults $P^{(t)}$.
\end{proof}

\section{Quantum Gabidulin codes}
\label{sec:quantum_gabidulin_codes}

\cref{lemma:w_stacked_faults_induce_rank_4w_change} shows that a small number of faults induce a low-rank error on the output state of the stacked memory. Following our toy model, we introduce a quantum generalization of Gabidulin codes to correct these low-rank errors.

The field $\FQ$ with $Q = 2^n$ is a $n$-dimensional space over $\Ftwo$ equipped with the $\Ftwo$-valued trace 
$
\tr{\alpha} := \alpha + \alpha^2 + \dots + \alpha^{2^{n-1}}.
$
If $n$ is odd, it admits a {\em trace-orthogonal normal} basis (\cite{macwilliams1977theory} Ch. 4 \S 9), that is a basis of the form 
$
\balpha = [\alpha,\alpha^2,\ldots,\alpha^{2^{n-1}}]
$ for some $\alpha \in \FQ$ with 
\begin{align}
\label{eq:trace_orthonomality_for_powers_of_alpha}
\tr{\alpha^{2^i} \alpha^{2^j}} = \delta_{i, j}
\end{align}
for all $i,j = 1,\dots, n$.
In the remainder of this paper, $n$ is odd and $\balpha$ is a trace-orthogonal normal basis.

The $\FQ$-linear space $\FQ^n$ is equipped with the inner product $\inner{\bbeta}{\bgamma} := \sum_{i=1}^n \tr{\beta_i \gamma_i}$
where 
$\bbeta = [\beta_1, \dots, \beta_n]$ and 
$\bgamma = [\gamma_1, \dots, \gamma_n]$ are in $\FQ^n$.

Consider the vectors 
\begin{center}
\begin{tabular}{cccccccc}
$\balpha $
& $= [$
& $\alpha,$
& $\alpha^2,$
& $\dots$
& $\alpha^{2^{n-2}},$
& $\alpha^{2^{n-1}}$
& ],
\\
$\balpha^2 $
& $= [$
& $\alpha^2,$
& $\alpha^4,$
& $\dots$
& $\alpha^{2^{n-1}},$
& $\alpha$
& ],
\\
&
$\vdots$
&&&&&&
\\
$\balpha^{2^n-1} $
& $= [$
& $\alpha^{2^{n-1}},$
& $\alpha,$
& $\dots$
& $\alpha^{2^{n-3}},$
& $\alpha^{2^{n-2}}$
& ].
\\
\end{tabular}
\end{center}

\begin{proposition}
\label{prop:gabidulin_code_dual}
The dual code of $\gab{\balpha, r}$ is $\gab{\balpha^{2^r}, n-r}$.
\end{proposition}

\begin{proof}
By definition, $\gab{\balpha, r}$ is generated by the vectors $\balpha^{2^{i}}$ with $i < r$ and the
$
\balpha,\balpha^2,\ldots,\balpha^{2^{n-1}}
$
form an orthonormal basis of $\FQ^n$.
Therefore, its dual is generated by the vectors $\balpha^{2^{r}}, \balpha^{2^{r+1}}, \dots, \balpha^{2^{n-1}}$. It is the code $\gab{\balpha^{2^r}, n-r}$.
\end{proof}

Any $\bbeta \in \FQ^n$ can be interpreted as a $n \times n$ matrix by replacing each $\beta_i \in \FQ$ by the column vector obtained by expressing $\beta_i$ in the basis $\balpha$.
Denote by $X(\beta)$ the $X$ type Pauli operator acting on a $n \times n$ stacked memory whose support is given by the matrix representation of $\bbeta$.
The $Z$ type Pauli error $Z(\bbeta)$ is defined similarly.

\begin{lemma}
\label{lemma:commutation_from_orthogonality}
Let $\bbeta, \bgamma$ in $\FQ^n$. Then $X(\bbeta)$ and $Z(\bgamma)$ commute if and only if $\inner{\bbeta}{\bgamma} = 0$.
\end{lemma}

\begin{proof}
In the basis $\balpha$, we have 
$\beta_j = \sum_{i=1}^n b_{i, j} \alpha^{2^i}$ 
and 
$\gamma_j = \sum_{i'=1}^n c_{i', j} \alpha^{2^{i'}}$
with $b_{i, j}, c_{i', j}$ in $\Ftwo$.
Their inner product is
\begin{align*}
\inner{\bbeta}{\bgamma}
= \sum_{j=1}^n \sum_{i=1}^n \sum_{i'=1}^n b_{i,j} c_{i',j} \tr{\alpha^{2^i}\alpha^{2^{i'}}}
= \sum_{j=1}^n \sum_{i=1}^n b_{i,j} c_{i,j}
\end{align*}
where the second equality is derived from \cref{eq:trace_orthonomality_for_powers_of_alpha}.
The lemma follows from this expression.
\end{proof}

Let $r, s$ be integers such that $r+s < n$. The {\em quantum Gabidulin code} $\qgab{\balpha, r, s}$ is the stabilizer code defined by the stabilizer generators 
$
X(\bbeta)
$
with $\bbeta \in \gab{\balpha, r}$
and 
$
Z(\bgamma)
$
with $\bgamma \in \gab{\balpha^r, s}$.
These operators commute based on \cref{lemma:commutation_from_orthogonality} and 
\cref{prop:gabidulin_code_dual}, because $\gab{\balpha^r, s}$ is included in the dual $\gab{\balpha^r, n-r}$ of $\gab{\balpha^r, r}$.

Define the {\em minimum rank distance} of a quantum Gabidulin code to be the minimum rank of a $n \times n$ Pauli error (for the notion of rank defined in \cref{sec:stacked_implementation}) which commutes with all the stabilizer and which is not a stabilizer (up to a global phase).
Recall that the stabilizers are the products of stabilizer generators.

\begin{theorem}
\label{theorem:quantum_gabidulin_codes_parameters}
For $r<n/2$, the quantum code $\qgab{\alpha,r,r}$ encodes $k=n(n-2r)$ logical qubits into $n^2$ physical qubits and its minimum rank distance at least $r+1$.
\end{theorem}

\begin{proof}
The number of independent $X$ stabilizer generators is the dimension of $\gab{\balpha, r}$ over $\Ftwo$, that is $nr$.
The same argument for $Z$ stabilizer generators leads to $k = n^2 - 2nr$.

The minimum rank is reached either for an $X$ error or for a $Z$ error.
Any $Z$ error commuting with all the stabilizers corresponds to a vector of $\gab{\balpha, r}^\perp$, which is $\gab{\balpha^{2^r}, n-r}$ by \cref{prop:gabidulin_code_dual}. If it is non-trivial its minimum rank is $r+1$ by \cref{theorem:gabidulin_codes_parameters}.
Similarly, a non-trivial $X$ error commuting with all the stabilizers corresponds to a vector of 
$
\gab{\balpha^{2^r}, r}^\perp = 
\gab{\balpha^{2^{2r}}, n-r}
$ 
and its rank is also lower bounded by $r+1$.
\end{proof}

To correct faults in the stacked implementation of a $n$-qubit Clifford circuit with $n$ layers, we encode the input state of the stacked memory using a quantum Gabidulin code $Q = \qgab{\balpha, r}$.
The output state of the stacked memory is then encoded in a different stabilizer code $Q'$ because the circuit $C$ was applied to each layer of the stacked memory. 
However, this transformation preserves the minimum rank distance $d = r+1$.
To correct an error $E$ on the output state of the stacked memory, we measure the stabilizer of $Q'$.
The outcome is the so called {\em syndrome}.
Then, we execute a {\em decoder} which uses the syndrome to determine a minimum rank Pauli correction $\hat E$ to apply to the stacked memory. 
If the syndrome can be measured accurately and if we have an efficient decoder for $Q'$, this approach corrects any combination of up to $r/8$ faults.

\section{Conclusion}

We introduced a quantum generalization of Gabidulin codes and we proposed a protocol for the correction of quantum circuit faults based on these codes.
This protocol may find applications to improve stabilizer state factories or magic state factories that are made with Clifford circuits (where only the input state is non-Clifford)~\cite{bravyi2005universal}.
It could play a role in preparing the stabilizer states consumed by CliNR~\cite{delfosse2024low}.
It may also be relevant in the Clifford part of variational quantum algorithms to probe multiple initial states in parallel~\cite{cerezo2021variational, peruzzo2014variational, anand2023hamiltonians}.

Further work is needed to make our approach practical.
First, we need a platform hosting multi-qubit cells with little crosstalk between cells.
Second, it is not realistic to assume that the syndrome can be measured exactly. A fault-tolerant syndrome extraction circuit is needed in practice. One could consider generalizing classical rank-metric LDPC codes~\cite{aragon2019low} to the quantum setting to make the syndrome extraction less noisy.
Third, a fast and efficient decoder must be designed for the output rank-metric code. LDPC codes could also make this task easier.
Fourth, one need to design magic state factories for stacked memories, in order to make this scheme universal for quantum computing.

Finally, we leave the question of building broader families of quantum rank-metric codes, providing more flexibility in the parameter choice, for future work.

\bibliography{references}

\appendix

\section{Proof of \cref{theorem:gabidulin_codes_parameters}}
\label{appendix:proof_gab_code_parameters}

For the sake of completeness, we include here a proof of this standard result~\cite{delsarte1978,gabidulin1985}.

\begin{proof}
Let $\bx_f=[f(\alpha_1),\ldots,f(\alpha_n)]$ be a codeword.
Its rank is equal to the dimension over $\Ftwo$ of the $\Ftwo$-subspace $V$ of $\FQ$ generated
by $f(\alpha_1),\ldots,f(\alpha_n)$.
Since $\alpha_1,\ldots,\alpha_n$ is a basis
of $\FQ$ the dimension of $V$ is just the dimension of the image of $f$ in $\FQ$, which is therefore equal to $n-\dim\ker f$. But the number of elements in $\ker f$ is at most $2^{k-1}$ since the kernel of $f$ is the set of roots in $\FQ$ of a polynomial of degree at most $2^{k-1}$. Therefore $\dim\ker f\leq
k-1$ and the rank of a non-zero vector of the code is therefore at least $n-k+1$.
The minimum rank of a non-zero vector is exactly $n-k+1$ by the Singleton bound.
The same argument implies in particular that the map
$
f \mapsto [f(\alpha_1),\ldots,f(\alpha_n)],
$
where $f$ ranges over the polynomials \eqref{eq:qpoly}, has zero kernel:
therefore the dimension of the code equals that of the space of polynomials, namely $k$.
\end{proof}

\end{document}